\documentclass[preprint,12pt,authoryear]{elsarticle}
\usepackage[T1]{fontenc}
\usepackage{lmodern,amsmath,amssymb,amsthm,mathtools,booktabs}
\usepackage{algorithm,algpseudocode,microtype,setspace,xurl}
\usepackage[a4paper,margin=25mm]{geometry}
\usepackage[hidelinks]{hyperref}
\newtheorem{theorem}{Theorem}[section]
\newtheorem{lemma}[theorem]{Lemma}
\newtheorem{proposition}[theorem]{Proposition}
\newtheorem{corollary}[theorem]{Corollary}
\theoremstyle{remark}

\newcommand{\OPT}{\operatorname{OPT}}
\newcommand{\V}{\mathcal V}
\newcommand{\A}{\mathcal A}
\newcommand{\Sset}{\mathcal S}
\newcommand{\U}{\widehat U}

\newcommand{\Zplus}{\mathbb Z_{\ge0}}
\DeclareMathOperator*{\argmin}{arg\,min}
\journal{European Journal of Operational Research}
\begin{document}
\begin{frontmatter}
\title{Critical-state algorithms for lot sizing with all-units discounts and variable inventory bounds}
\author{Kleitos Papadopoulos}
\begin{abstract}
We study deterministic single-item lot sizing with all-units quantity discounts,
linear holding costs, and time-varying end-of-period inventory bounds.
Unit prices may vary arbitrarily over time. For one common positive discount
threshold, we give an exact algorithm requiring $O(n^3)$ arithmetic operations
and $O(n^2)$ working space, including recovery of an optimal purchase plan.
The method combines a critical-state theorem in cumulative-purchase coordinates
with two prefix-minimum scans per dynamic-programming layer. The structural
argument accommodates downward price discontinuities and distinguishes
prescribed from free terminal inventory. With a fixed global set of $m$ positive
threshold values, the bounds become $O_m(n^{m+2})$ time and $O_m(n^{m+1})$
space. Period-specific thresholds admit the same exact construction, but its
state set can be exponential. For commensurate flow data, an inventory-lattice
implementation takes $O(PW)$ time, where $P$ is the total number of price tiers
and $W$ is the largest inventory-lattice width. We also derive a conditional
optimality certificate from a monotone one-threshold relaxation; this certificate
does not imply a near-linear worst-case algorithm for arbitrary prices.
Exact-arithmetic implementations are checked against exhaustive inventory
dynamic programs and tier-enumeration linear programs. An illustrative timing
study identifies settings in which the lattice implementation is preferable.
\end{abstract}
\end{frontmatter}

\section{Introduction}\label{sec:intro}
An all-units quantity discount applies a lower unit price to an entire order once
its quantity reaches a threshold. Unlike an incremental discount, the resulting
procurement cost can decrease discontinuously as the order grows. In a
multi-period setting, a discount therefore creates a trade-off among ordering
more than immediate demand, paying holding costs, and occupying limited storage
space. When prices fluctuate over time, the cheapest inventory remaining in a
particular policy need not identify the cheapest complete history for another
inventory level. Exact methods must compare total costs while preserving both
threshold eligibility and the storage constraints of each history.

This paper considers continuous order quantities, nonnegative deterministic
demands, linear nonnegative holding costs, and arbitrary finite end-of-period
inventory bounds. There are no setup charges, separate production limits,
backlogging, or disposal. The entire horizon is known in advance. Our main
algorithm does not require either unit-price sequence to be monotone. Its
central observation is that an optimal plan can be selected from a finite set
of cumulative-purchase levels obtained by translating inventory-bound values
by sums of discount thresholds. For one common threshold, this set has
$O(n^2)$ elements. The apparent quadratic predecessor search in each state
then reduces to two prefix minima, giving an $O(n^3)$ algorithm.

\subsection{Related literature and positioning}
Dynamic lot sizing has a long algorithmic history beginning with the classical
model of \citet{wagner1958}. Boundary-based descriptions of optimal production
and inventory trajectories appear in \citet{love1973} and
\citet{swoveland1975}. Our proof uses this general extreme-point perspective,
but expresses it directly in cumulative-purchase coordinates for the stated
all-units model.

For piecewise-concave production costs, \citet{koca2014} establish polynomial
solvability with a fixed number of common breakpoints, including an
$O(n^{2m+3})$ bound for their general no-backlogging formulation. Their
extensions also use stock-boundary states. In related piecewise-linear models,
\citet{ou2017} obtain an $O(n^{m+2}\log n)$ algorithm, and
\citet{pappl2026} report an $O(n^{m+2})$ algorithm. The latter work was initially
circulated as arXiv:2403.16314 and has since appeared in the \emph{Journal of
Global Optimization}. Thus the exponent $m+2$, and polynomial solvability for
fixed common breakpoints, are not claimed here to be new in themselves.

\citet{down2021} study a bounded-inventory one-breakpoint special case and give
a quadratic algorithm based on a just-in-time structural property.
\citet{papmono2026} state an $O(n\log n)$ algorithm for the one-breakpoint
case with non-increasing price sequences and variable inventory capacities.
That result is relevant to the conditional certificate in
Section~\ref{sec:certificate}; our main correctness theorem does not rely on it.
For broader piecewise-linear production models,
\citet{shaw1998} develop a pseudopolynomial algorithm. This provides context
for our lattice implementation, without identifying production-capacity and
inventory-capacity formulations.

A recent preprint by \citet{papmonge2026} develops Monge matrix searching for
nonlinear piecewise-concave production costs, using common production
breakpoints and a fixed global collection of stock-boundary levels.
The present method instead uses affine tier costs and arbitrary period-specific
inventory bounds: their cumulative boundary values contribute only $O(n)$
anchors, without assuming a fixed number of distinct capacity values.
We use elementary prefix minima rather than Monge searching. These differences
specify the scope of the contribution; they are not a claim that the present
algorithm dominates every algorithm for the broader models.

\subsection{Results and organization}
We give a self-contained critical-state proof, an implementable cubic
one-threshold algorithm, and a space-efficient recovery procedure. The proof
extends to a fixed global set of threshold values and to unrestricted local
thresholds, with explicit accounting for the resulting state-set sizes. A
second implementation exploits a manageable inventory lattice. Finally, a
historical-price relaxation provides a certified shortcut when a reconstructed
plan attains its lower bound. The shortcut is optional, and a rejected
certificate leaves the exact dynamic program unchanged.

Section~\ref{sec:model} defines the model and its cumulative coordinates.
Section~\ref{sec:single} proves the cubic result.
Section~\ref{sec:multi} treats multiple and period-dependent thresholds.
Section~\ref{sec:lattice} gives the lattice algorithm.
Section~\ref{sec:certificate} develops the monotone-relaxation certificate.
Section~\ref{sec:validation} reports reproducible validation and an illustrative
timing study. Section~\ref{sec:conclusion} discusses scope and further work.
\ref{app:reconstruction} supplies the full reconstruction argument.

\section{Model and preliminary reductions}\label{sec:model}
\subsection{Inventory and procurement}
There are $n\ge1$ periods. Period $t$ has demand $d_t\ge0$, end-of-period
inventory bound $B_t\ge0$, and holding cost $h_t\ge0$ per unit of ending
inventory. All data are finite. An order $x_t\ge0$ and inventory $I_t$ satisfy
\begin{align}
 I_0&=0, & I_t&=I_{t-1}+x_t-d_t, & 0&\le I_t\le B_t
 \qquad(t=1,\ldots,n).\label{eq:inventory}
\end{align}
For the one-threshold model, $Q>0$ is common to all periods and
$0\le p_{2,t}\le p_{1,t}$. The procurement cost is
\begin{equation}
 f_t(x)=\begin{cases}
 p_{1,t}x,&0\le x<Q,\\
 p_{2,t}x,&x\ge Q.
 \end{cases}\label{eq:onecost}
\end{equation}
Both price sequences are otherwise arbitrary. We minimize
\begin{equation}
 C(x)=\sum_{t=1}^n\bigl(f_t(x_t)+h_tI_t\bigr).
 \label{eq:objective}
\end{equation}
Orders are continuous rather than restricted to integer units. In particular,
no rounding of quantities or thresholds is part of the critical-state method.

The bound $B_t$ applies \emph{after} period-$t$ demand has been met. There is no
additional restriction on $I_{t-1}+x_t$. A model whose post-order inventory
cannot exceed $T_t$ is obtained by setting $B_t=T_t-d_t$, provided
$T_t\ge d_t$. Conversely, an end-inventory instance can be represented using
post-order bounds $T_t=B_t+d_t$. An independent upper bound on $x_t$ is a
different constraint and is not included in the main theorem.

We treat two terminal policies: prescribed inventory
$I_n=b$, where $0\le b\le B_n$, and free inventory $I_n\in[0,B_n]$.
These policies cannot be interchanged. For example, with one period,
$d_1=1$, $B_1=4$, $Q=4$, prices $(10,1)$, and no holding cost,
the zero-terminal optimum costs $10$, whereas the free-terminal optimum
orders four units and costs $4$.

\subsection{Cumulative-purchase coordinates}
Define cumulative demand and cumulative purchasing by
\begin{equation}
 D_0=Y_0=0,\qquad D_t=\sum_{s=1}^t d_s,\qquad
 Y_t=\sum_{s=1}^t x_s.\label{eq:cumulative}
\end{equation}
Then $I_t=Y_t-D_t$ and $x_t=Y_t-Y_{t-1}$. Feasibility becomes
\begin{equation}
 Y_0=0,\qquad Y_t\ge Y_{t-1},\qquad D_t\le Y_t\le D_t+B_t.
 \label{eq:corridor-original}
\end{equation}
Set $U_t=D_t+B_t$. Under prescribed terminal inventory, replace $U_n$ by
$D_n+b$ and retain the equality $Y_n=D_n+b$. Under free terminal inventory,
no replacement is made. Put
\begin{equation}
 \U_t=\min_{s=t,\ldots,n}U_s,\qquad H=\U_n,
 \qquad \A=\{0\}\cup\{D_t,\U_t:1\le t\le n\}.
 \label{eq:anchors}
\end{equation}
The number of distinct anchors is at most $2n+1$.

\begin{lemma}[Suffix tightening]\label{lem:tightening}
Replacing the upper bounds in~\eqref{eq:corridor-original} by $\U_t$ preserves
the set of complete feasible plans, under either terminal policy.
\end{lemma}
\begin{proof}
For $s\ge t$, nonnegative orders imply $Y_t\le Y_s\le U_s$.
Taking the minimum over $s$ gives $Y_t\le\U_t$. Conversely,
$\U_t\le U_t$, so every trajectory satisfying the tightened bounds satisfies
the original ones. The terminal equality, when present, is retained.
\end{proof}

All feasible coordinates and anchors lie in $[0,H]$. The array $\U$ is
computed by one backward scan. The tightening is a property of complete
trajectories, not a capacity allowance for modifying an arbitrary earlier
policy. It is legitimate for an offline dynamic program even though some
prefixes excluded by the tightening would be feasible in isolation.

\subsection{Computational model}
We count exact arithmetic operations, comparisons, and array accesses.
The notation $O_m(\cdot)$ permits constants depending on a fixed number $m$
of distinct positive threshold values. It does not denote fixed-parameter
tractability, and $m$ is not allowed to grow silently with $n$.
For rational inputs, the accompanying implementations use exact fractions;
the displayed bounds exclude the bit cost of arithmetic. No constant-time
representation or comparison of arbitrary symbolic real expressions is assumed.

\section{A cubic algorithm for one common threshold}\label{sec:single}
\subsection{An eligible-line representation}
The discount discontinuity can be handled without open feasible regions.
For every $x\ge0$,
\begin{equation}
 f_t(x)=\min\left\{p_{1,t}x,\quad
       p_{2,t}x\;\text{subject to }x\ge Q\right\}.
 \label{eq:eligible-one}
\end{equation}
An ineligible term has value $+\infty$. The regular-price line may be used
above $Q$ in this representation because it is then dominated by the
discounted-price line for the same order. The lower threshold of the latter
line cannot be dropped.

For each vector $z\in\{0,1\}^n$, let $\mathcal P_z$ be the tightened
inventory polytope with additional inequalities
\begin{equation}
 Y_t-Y_{t-1}\ge z_tQ.\label{eq:tier-lp}
\end{equation}
Minimize on $\mathcal P_z$ the affine objective
\begin{equation}
 L_z(Y)=\sum_{t=1}^n\left[
 ((1-z_t)p_{1,t}+z_tp_{2,t})(Y_t-Y_{t-1})
       +h_t(Y_t-D_t)\right].\label{eq:affine-one}
\end{equation}
These polyhedra may overlap. They are used for the proof, not enumerated
by the algorithm.

\begin{lemma}[Affine selection]\label{lem:affine-selection}
The original optimum is the minimum of the feasible linear-programming
optima $\min_{Y\in\mathcal P_z}L_z(Y)$. An optimal vertex of a best such
linear program is an optimum of the original problem.
\end{lemma}
\begin{proof}
For each feasible $Y$, identity~\eqref{eq:eligible-one} allows a minimizing
tier to be selected independently in every period. Hence the pointwise
minimum of the eligible affine objectives equals $C(x)$. There are finitely
many tier selections, and every nonempty $\mathcal P_z$ is compact because
$0\le Y_t\le H$. Its affine objective attains a minimum at a vertex.
The minimum over all tier selections therefore equals the original optimum.
At a minimizing vertex $v$ of a best selection, $C(v)\le L_z(v)=\OPT$;
feasibility gives the reverse inequality.
\end{proof}

\subsection{Critical cumulative-purchase levels}
Define
\begin{equation}
 \V=\{a+kQ:a\in\A,\ k\in\mathbb Z,\ -n\le k\le n\}\cap[0,H].
 \label{eq:grid-one}
\end{equation}
After duplicate values are removed, let $M=|\V|$. Then
$M\le(2n+1)^2$.

\begin{theorem}[One-threshold critical states]\label{thm:critical-one}
Under either terminal policy, there is an optimal plan with $Y_t\in\V$
for all $t=0,\ldots,n$.
\end{theorem}
\begin{proof}
Choose the optimal vertex supplied by Lemma~\ref{lem:affine-selection}.
The possible active equalities are
\begin{equation}
 Y_t=D_t,\quad Y_t=\U_t,\quad
 Y_t-Y_{t-1}=0,\quad Y_t-Y_{t-1}=Q,
 \label{eq:active-one}
\end{equation}
together with the fixed initial coordinate and any prescribed terminal
coordinate. The latter equals $U_n=\U_n$ and is an anchor.

Consider the path on nodes $0,\ldots,n$. Retain the edge $(t-1,t)$ whenever
its difference inequality is active. Mark a node as anchored if its
coordinate is fixed by an inventory bound, the initial condition, or the
terminal equality. Every connected component of retained edges has an
anchored node. Otherwise all coordinates in that component could be shifted
by the same sufficiently small positive or negative amount. Active
difference equalities would remain unchanged, and all remaining inequalities
would remain feasible because they have positive slack. This would express
the purported vertex as the midpoint of two distinct feasible points.

Starting at an anchored value $a$, traverse retained edges to any other
node in its component. Each edge changes the coordinate by $0$ or $Q$.
Traversing forward adds these increments; traversing backward subtracts
them. The path contains at most $n$ edges, so the resulting coordinate
is $a+kQ$ with $|k|\le n$. Feasibility also places it in $[0,H]$.
\end{proof}

The grid contains exact real quantities determined by the input. Its size
does not depend on the numerical magnitude of demand, capacity, or $Q$.
The theorem does not say that every unrestricted optimal prefix ending at
a grid point can be represented on this same grid. Only the existence of
a grid-restricted globally optimal complete trajectory is required.

\subsection{Two prefix minima per layer}
Sort the grid as $v_1<\cdots<v_M$. Let $G_t(v)$ denote the minimum cost of
a path restricted to this grid up to boundary $t$, with $Y_t=v$.
Initialize $G_0(0)=0$ and all other entries to $+\infty$. For a feasible
target $D_t\le v\le\U_t$, the Bellman equation is
\begin{align}
 G_t(v)=h_t(v-D_t)+\min\biggl\{&
 p_{1,t}v+\min_{u\in\V:\,u\le v}
                 \bigl(G_{t-1}(u)-p_{1,t}u\bigr),\notag\\[-2pt]
 &p_{2,t}v+\min_{u\in\V:\,u\le v-Q}
                 \bigl(G_{t-1}(u)-p_{2,t}u\bigr)\biggr\}.
 \label{eq:dp-one}
\end{align}
Set $G_t(v)=+\infty$ outside the feasible band. The first term permits
all nonnegative order quantities; the second requires an order of at least
$Q$. Their minimum is the exact all-units cost by~\eqref{eq:eligible-one}.

For the first term, scan targets in increasing order and maintain the
smallest value of $G_{t-1}(u)-p_{1,t}u$ in the prefix $u\le v$.
For the second, advance a pointer through the prefix $u\le v-Q$ and maintain
the corresponding minimum. Each predecessor enters each prefix once.
Both scans read the same frozen previous layer; they do not use values
computed earlier in the current layer. Hence one period costs $O(M)$
time and two arrays suffice for objective computation.

\begin{theorem}[Cubic one-threshold algorithm]\label{thm:cubic}
The one-common-threshold problem with arbitrary temporal prices and
variable end-inventory bounds can be solved exactly in $O(n^3)$ arithmetic
operations and $O(n^2)$ working space. These bounds include recovery of
an optimal purchase plan.
\end{theorem}
\begin{proof}
The Bellman equation compares all feasible grid predecessors and charges
exact procurement and holding costs. It therefore computes the optimum
among grid-restricted paths. For prescribed terminal inventory, return
$G_n(D_n+b)$; for free terminal inventory, return $\min_{v\in\V}G_n(v)$.
Theorem~\ref{thm:critical-one} shows that this final value equals the
unrestricted optimum.

Constructing and sorting~\eqref{eq:grid-one} takes $O(n^2\log(n+1))$
operations and $O(n^2)$ space. The $n$ layers require $O(nM)=O(n^3)$
operations and $O(M)$ space. \ref{app:reconstruction} proves
that a forward--backward midpoint reconstruction takes $O(nM)$ further
operations and $O(M+n)$ space for two tiers. Its additional time and space
fit within the displayed bounds.
\end{proof}

\subsection{A small instance}
Consider $Q=4$, zero holding costs, and the data in
Table~\ref{tab:example-one}, with $I_3=0$.

\begin{table}[htbp]\centering
\caption{An arbitrary-price instance with an optimal cost of $15$.}
\label{tab:example-one}
\begin{tabular}{rrrrrrr}
\toprule
$t$&$d_t$&$B_t$&$p_{1,t}$&$p_{2,t}$&$x_t^*$&$I_t^*$\\
\midrule
1&1&4&5&1&5&4\\
2&2&8&10&3&0&2\\
3&4&8&5&4&2&0\\
\bottomrule
\end{tabular}
\end{table}

The optimal cumulative coordinates are $(Y_1,Y_2,Y_3)=(5,5,7)$.
Here $D=(1,3,7)$ and the tightened upper bounds are $(5,7,7)$,
so these coordinates belong to the critical grid. The first order uses
the first-period discount, while the last order uses the regular tier.
Neither sorting periods by price nor comparing only the price of surviving
inventory describes this policy. Equation~\eqref{eq:dp-one} evaluates complete
histories and their capacities instead.

\section{Multiple and period-dependent thresholds}\label{sec:multi}
\subsection{General all-units pricing}
In period $t$, let
\begin{equation}
 0=q_{t,0}<q_{t,1}<\cdots<q_{t,m_t}<q_{t,m_t+1}=+\infty,
 \qquad p_{t,0}\ge\cdots\ge p_{t,m_t}\ge0.
 \label{eq:thresholds}
\end{equation}
The cost is $f_t(x)=p_{t,k}x$ for $q_{t,k}\le x<q_{t,k+1}$.
Equality at a threshold receives the newly eligible lower price.
Prices are non-increasing with quantity within each period, but are
unrestricted over time. Define the total tier count and maximum local
count by
\begin{equation}
 P=\sum_{t=1}^n(m_t+1),\qquad K_{\max}=\max_t(m_t+1).
 \label{eq:tiercounts}
\end{equation}
The one-threshold model is recovered by $m_t=1$ and $q_{t,1}=Q$,
with $p_{t,0}=p_{1,t}$ and $p_{t,1}=p_{2,t}$.

The eligible-line representation generalizes to
\begin{equation}
 f_t(x)=\min_{k:\,q_{t,k}\le x}p_{t,k}x.\label{eq:eligible-multi}
\end{equation}
The largest eligible index has a weakly smaller unit price than all
previous ones. Accordingly, every line can extend past its nominal upper
endpoint, but its lower eligibility threshold must remain.

\subsection{A critical grid for local thresholds}
Retain $D$, $\U$, $H$, and $\A$ from Section~\ref{sec:model}. Define
\begin{align}
 \Sset&=\left\{\sum_{t=1}^n z_t:
          z_t\in\{0,q_{t,1},\ldots,q_{t,m_t}\}\right\}\cap[0,H],
          \label{eq:sums-local}\\
 \V&=\{a+s,a-s:a\in\A,\ s\in\Sset\}\cap[0,H].
          \label{eq:grid-multi}
\end{align}
Write $R=|\Sset|$ and $M=|\V|$. Notice that $R\le M$ because $0$ is an
anchor and every $s\in\Sset$ is a grid point. Also,
$M\le2(2n+1)R$.

\begin{theorem}[Multi-threshold critical states]\label{thm:critical-multi}
Under either terminal policy, there is an optimal plan whose cumulative
purchases all belong to~\eqref{eq:grid-multi}.
\end{theorem}
\begin{proof}
For each tier selection $(k_1,\ldots,k_n)$, minimize the affine objective
\[
 \sum_{t=1}^n\left[p_{t,k_t}(Y_t-Y_{t-1})+h_t(Y_t-D_t)\right]
\]
on the tightened inventory polytope with additional inequalities
$Y_t-Y_{t-1}\ge q_{t,k_t}$. Identity~\eqref{eq:eligible-multi} and the
argument of Lemma~\ref{lem:affine-selection} give a globally optimal vertex
of a best selection.

Retain active edges $(t-1,t)$ exactly as in
Theorem~\ref{thm:critical-one}. Their increments are now the selected
local thresholds, including zero. Every retained-edge component has an
anchor. A path from an anchor to a coordinate adds, or subtracts, at most
one local threshold from each intervening period. Choose zero in every
other period. The resulting sum belongs to the untruncated set
in~\eqref{eq:sums-local}. Both the anchor and the coordinate lie in $[0,H]$,
so their distance is at most $H$; the sum is therefore retained in $\Sset$.
\end{proof}

The use of $a+s$ and $a-s$, rather than arbitrary signed sums, follows from
the path geometry: all edges on a path from one time to another are
traversed in the same temporal direction. The local-choice construction
also avoids allowing a single period to contribute two distinct thresholds.
It may still retain many values that are unnecessary for a particular
instance, but it is guaranteed not to exclude every optimal plan.

\subsection{Tierwise prefix-minimum dynamic programming}
For each feasible target $D_t\le v\le\U_t$,
\begin{equation}
 G_t(v)=h_t(v-D_t)+\min_{0\le k\le m_t}\left\{
 p_{t,k}v+\min_{u\in\V:\,u\le v-q_{t,k}}
          \bigl(G_{t-1}(u)-p_{t,k}u\bigr)\right\}.
 \label{eq:dp-multi}
\end{equation}
The initialization and terminal selection are unchanged. Every tier
requires one prefix-minimum scan; there is no optimization over the
nominal upper endpoint of that tier. The total time for value layers is
$O(PM)$ once the grid is sorted.

\begin{algorithm}[htbp]
\caption{One period of the critical-grid dynamic program}
\label{alg:layer}
\begin{algorithmic}[1]
\Require Sorted grid $v_1<\cdots<v_M$; previous values $A_1,\ldots,A_M$;
period $t$
\Ensure Current values $C_1,\ldots,C_M$
\State $C_i\gets+\infty$ for all $i$
\For{$k=0,\ldots,m_t$}
 \State $j\gets1$; $\beta\gets+\infty$
 \For{$i=1,\ldots,M$}
  \While{$j\le M$ and $v_j\le v_i-q_{t,k}$}
   \State $\beta\gets\min\{\beta,A_j-p_{t,k}v_j\}$
   \State $j\gets j+1$
  \EndWhile
  \State $C_i\gets\min\{C_i,p_{t,k}v_i+\beta\}$
 \EndFor
\EndFor
\For{$i=1,\ldots,M$}
 \If{$D_t\le v_i\le\U_t$}
  \State $C_i\gets C_i+h_t(v_i-D_t)$
 \Else
  \State $C_i\gets+\infty$
 \EndIf
\EndFor
\end{algorithmic}
\end{algorithm}

\begin{proposition}[Correctness on the finite grid]\label{prop:grid-dp}
Algorithm~\ref{alg:layer}, applied successively with~\eqref{eq:grid-multi},
returns the unrestricted global optimum under the specified terminal policy.
\end{proposition}
\begin{proof}
For fixed predecessor $u$ and target $v$, $x_t=v-u\ge q_{t,k}$ is equivalent
to $u\le v-q_{t,k}$. Minimizing over eligible tiers gives
$f_t(v-u)$ exactly. Therefore~\eqref{eq:dp-multi} is the ordinary shortest-path
recurrence on the layered grid. Its final value is the restricted optimum.
Theorem~\ref{thm:critical-multi} supplies an unrestricted optimum on that grid.
Conversely, every grid trajectory satisfies the original balance and
inventory constraints, so the restricted optimum cannot be smaller.
\end{proof}

\subsection{Constructing the grid and bounding the running time}
The sum set is generated recursively from $\Sset^{(0)}=\{0\}$:
\begin{equation}
 \Sset^{(t)}=\bigcup_{q\in\{0,q_{t,1},\ldots,q_{t,m_t}\}}
                (\Sset^{(t-1)}+q)\cap[0,H].\label{eq:sum-generation}
\end{equation}
Every offset is nonnegative. Once a partial sum exceeds $H$, it can never
contribute to a retained final sum, so clipping is exact. Each stage merges
at most $K_{\max}$ sorted translated lists and removes duplicates. Since
zero is always an option, every retained partial sum remains in the final
set, and each partial set has size at most $R$. A standard heap-based
multiway merge therefore constructs $\Sset$ in
$O(PR\log(K_{\max}+1))$ operations.

For each anchor, the lists of feasible $a+s$ values and of feasible $a-s$
values, with the latter traversed in reverse sum order, are sorted.
Merging these $O(n)$ lists takes $O(nR\log(n+1))$ operations. This construction
uses comparisons rather than an assumed constant-time hashing operation.

\begin{theorem}[General instance-dependent bounds]\label{thm:general}
For arbitrary local thresholds, the optimal value can be computed in
\begin{equation}
 O\bigl(PR\log(K_{\max}+1)+nR\log(n+1)+PM\bigr)
 \label{eq:general-time}
\end{equation}
operations. Working space, excluding input storage, is $O(M+n)$.
An optimal plan can be recovered within the same working-space bound
using at most $O(nK_{\max}M)$ additional operations.
\end{theorem}
\begin{proof}
The preceding merges account for preprocessing. The value layers account
for $PM$ operations. The sum and grid arrays, merge heaps, and two value
layers occupy $O(M+n)$ space because $R\le M$; moreover $K_{\max}\le R+1$
after irrelevant thresholds above $H$ are ignored. Processing the original
input costs $O(P)$, which is already covered. The reconstruction bounds
follow from \ref{app:reconstruction}.
\end{proof}

Suppose now that every positive threshold belongs to a fixed global set
$\{Q_1,\ldots,Q_m\}$. Periods may use different subsets of this set.
Every sum in~\eqref{eq:sums-local} has the form
\begin{equation}
 \sum_{r=1}^m k_rQ_r,\qquad k_r\in\Zplus,\qquad
 \sum_{r=1}^m k_r\le n.\label{eq:count-vector}
\end{equation}
Counting the nonnegative integer vectors yields
\begin{equation}
 R\le {n+m\choose m},\qquad
 M\le2(2n+1){n+m\choose m}.\label{eq:fixed-count}
\end{equation}
Different vectors may yield the same sum; the algorithm removes such
coincidences rather than counting them as different states.

\begin{corollary}[Fixed global threshold set]\label{cor:fixed-m}
For fixed $m\ge1$, the problem with arbitrary temporal prices and arbitrary
period-specific inventory bounds is solvable in $O_m(n^{m+2})$ arithmetic
operations and $O_m(n^{m+1})$ working space, including an optimal plan.
\end{corollary}
\begin{proof}
For fixed $m$, $R=O_m(n^m)$, $M=O_m(n^{m+1})$,
$P\le(m+1)n$, and $K_{\max}\le m+1$. Substitute these inequalities in
Theorem~\ref{thm:general}. The grid-construction term
$O_m(n^{m+1}\log(n+1))$ is dominated by $O_m(n^{m+2})$.
\end{proof}

\begin{table}[htbp]\centering
\caption{Bounds for a fixed global set of positive threshold values.
Working space includes plan recovery but excludes input storage and bit costs.}
\label{tab:complexities}
\begin{tabular}{lll}
\toprule
Global threshold values&Arithmetic time&Working space\\
\midrule
One common threshold&$O(n^3)$&$O(n^2)$\\
Two distinct values&$O(n^4)$&$O(n^3)$\\
Three distinct values&$O(n^5)$&$O(n^4)$\\
Fixed $m$ values&$O_m(n^{m+2})$&$O_m(n^{m+1})$\\
\bottomrule
\end{tabular}
\end{table}

\subsection{Why a local tier bound is not a global breakpoint bound}
With unrestricted period-dependent thresholds, only
\begin{equation}
 R\le\prod_{t=1}^n(m_t+1)\label{eq:product-bound}
\end{equation}
is guaranteed. Even one positive threshold per period can produce an
exponential set: choose $q_{t,1}=2^{t-1}$ and a final feasible cumulative
range containing $[0,2^n-1]$. All $2^n$ subset sums are distinct and survive
the truncation. This situation can occur, for example, when a final demand
of $2^n-1$ is preceded by zero demands and sufficiently large inventory
bounds. The example establishes exponential growth of this construction,
not NP-hardness of the exact restricted model.

Thus a bound on $m_t$ in every period does not justify treating $m$ as a
fixed constant in Corollary~\ref{cor:fixed-m}. Conversely, local thresholds
drawn from a fixed global set satisfy that corollary even when each period
uses a different subset. The distinction concerns threshold \emph{values},
not the number of rows in a price table.

\section{An inventory-lattice implementation}\label{sec:lattice}
Suppose demands, inventory bounds, all positive thresholds, and prescribed
terminal inventory are integer multiples of some $\delta>0$. Unit prices
and holding costs need not be multiples of $\delta$.

\begin{corollary}[Existence of a lattice optimum]\label{cor:lattice}
Under these assumptions, there is an optimal plan whose purchases and
inventories are all in $\delta\mathbb Z$.
\end{corollary}
\begin{proof}
Every cumulative-demand value, tightened upper bound, and anchor is a
multiple of $\delta$. Every local-choice sum is also such a multiple.
Theorem~\ref{thm:critical-multi} therefore gives a cumulative-purchase
trajectory on this lattice. Differences give lattice purchases and inventories.
\end{proof}

Let $F_t(i)$ be the best cost for end inventory
$i\in\{0,\delta,2\delta,\ldots,B_t\}$. The initialization is
$F_0(0)=0$. For every feasible lattice target $i$,
\begin{equation}
 F_t(i)=h_ti+\min_k\left\{p_{t,k}(d_t+i)+
 \min_{\substack{u\in\{0,\delta,\ldots,B_{t-1}\}\\
                  u\le d_t+i-q_{t,k}}}
       \bigl(F_{t-1}(u)-p_{t,k}u\bigr)\right\},\label{eq:lattice}
\end{equation}
where at $t=1$ the predecessor set consists only of zero. The condition
inside the minimum is precisely the tier eligibility condition for
$x_t=d_t+i-u$. Every tier again uses an expanding prefix.

\begin{proposition}[Lattice complexity]\label{prop:lattice-complexity}
With $W=1+\max_t B_t/\delta$, the optimal value can be computed in
$O(PW)$ operations and $O(W)$ working space. Storing one predecessor per
state yields an optimal plan with $O(nW)$ working space.
\end{proposition}
\begin{proof}
Every previous and current inventory layer has at most $W$ states.
A tier scans both layers once, giving $O(W)$ work, and the total tier count
is $P$. Corollary~\ref{cor:lattice} makes the final lattice optimum exact.
Predecessor indices give the stated direct recovery bound.
\end{proof}

This is a pseudopolynomial result: $W$ can be exponentially large in the
binary input length. It is nevertheless useful when stock capacities are
moderate in natural units. Choosing a coarser $\delta$ without checking that
all flow data are its multiples would invalidate the exactness claim.
The code raises an error rather than silently rounding incompatible data.
Unlike the cumulative coordinate, physical inventory need not be monotone;
we do not transfer the width-partition argument of
\ref{app:reconstruction} to this lattice implementation.

\section{A certificate from a monotone relaxation}\label{sec:certificate}
The cubic and fixed-$m$ algorithms are self-contained. An optional shortcut
can exploit a solver for a more structured relaxation. Throughout this
section, a \emph{monotone oracle} means an exact solver returning an optimal
purchase plan for the one-common-threshold model with non-increasing unit
prices, zero holding costs, and the specified inventory bounds. If this
oracle has $O(n\log n)$ arithmetic complexity, as stated in
\citet{papmono2026}, the accepted shortcut inherits that complexity.
This conditional use does not assume that the oracle can process an
arbitrary incoming value function or historical state updates.

\subsection{Removing holding costs and taking historical minima}
For the one-threshold instance, put
\begin{equation}
 \eta_t=\sum_{r=t}^n h_r,\qquad
 c_{\alpha,t}=p_{\alpha,t}+\eta_t,\qquad
 K=\sum_{t=1}^n h_tD_t\quad(\alpha=1,2).\label{eq:adjusted}
\end{equation}
Using $I_t=\sum_{s\le t}x_s-D_t$ gives
\begin{equation}
 \sum_t\bigl(f_t(x_t)+h_tI_t\bigr)
       =\sum_t\widetilde f_t(x_t)-K,\label{eq:holding-identity}
\end{equation}
where $\widetilde f_t$ has the same threshold and unit prices $c_{1,t},c_{2,t}$.
Now define
\begin{equation}
 \bar c_{\alpha,t}=\min_{j\le t}c_{\alpha,j},\label{eq:prefix-prices}
\end{equation}
and record the latest index $j_{\alpha,t}$ attaining each minimum.
The two sequences $\bar c_{\alpha,t}$ are non-increasing, with
$\bar c_{2,t}\le\bar c_{1,t}$. Replacing $c$ by $\bar c$ and setting holding
costs to zero relaxes the objective while retaining the feasible set.
If its optimal purchase cost is $R^*$, then
\begin{equation}
 L=R^*-K\le\OPT.\label{eq:lower-bound}
\end{equation}

\subsection{Whole-order reconstruction}
Let $z_t$ be an optimal relaxed order, and let $\alpha_t=1$ if $z_t<Q$
and $\alpha_t=2$ otherwise. Move each whole order to its recorded source:
\begin{equation}
 x'_j=\sum_{t:\,j_{\alpha_t,t}=j}z_t.\label{eq:reconstruction}
\end{equation}
Whole orders, rather than arbitrary fractions of discounted orders, are
moved so that the common threshold remains satisfied at the source.

\begin{lemma}[Cost of merging orders]\label{lem:subadditivity}
An all-units procurement cost with non-increasing unit prices satisfies
$f(a+b)\le f(a)+f(b)$ for all $a,b\ge0$. The same statement holds after
adding a nonnegative linear holding adjustment to every tier price.
\end{lemma}
\begin{proof}
For positive $a$ and $b$, the unit price at quantity $a+b$ is no greater
than the unit price at either constituent quantity. Multiplication by
$a$ and $b$ and addition gives the inequality. If one quantity is zero,
the claim follows from $f(0)=0$. Adding a linear term preserves it.
\end{proof}

The adjusted source cost of each moved order is its relaxed cost. Merging
orders at one source cannot increase that sum by
Lemma~\ref{lem:subadditivity}. Equation~\eqref{eq:holding-identity} then gives
$C(x')\le L$, where the expression $C$ can be evaluated even before
capacity feasibility is checked. Purchases only move backward, so
\begin{equation}
 \sum_{j\le t}x'_j\ge\sum_{j\le t}z_j\ge D_t.\label{eq:no-shortage}
\end{equation}
Total purchasing is unchanged. Thus shortages cannot be introduced and
terminal inventory is preserved, but an intermediate upper inventory bound
may be violated.

\begin{theorem}[One-threshold certificate]\label{thm:certificate-one}
If the reconstructed plan in~\eqref{eq:reconstruction} satisfies
$\sum_{j\le t}x'_j-D_t\le B_t$ for every $t$, it is globally optimal.
The certificate requires one monotone-oracle call and $O(n)$ additional
operations.
\end{theorem}
\begin{proof}
The reconstruction is feasible by the capacity test,
\eqref{eq:no-shortage}, and preservation of terminal inventory. Therefore
\[
 L\le\OPT\le C(x')\le L,
\]
so equality holds throughout. Prefix minima, source assignment, cumulative
inventory, and the capacity test all require linear time.
\end{proof}

For already non-increasing original prices and nonnegative holding costs,
$c_{\alpha,t}$ is also non-increasing. Choosing the latest minimizer yields
$j_{\alpha,t}=t$, so the certificate succeeds without moving any order.
For genuinely unlimited storage and prescribed terminal inventory, every
backward move is feasible and the certificate again succeeds. One can
represent this case with the redundant cumulative-purchase bound
$Y_t\le D_n+b$ before calling a finite-bound oracle.

With finite capacities, acceptance is not guaranteed. For example, take
$Q=4$, $d=(1,4)$, $B=(1,4)$, prices $p_1=(2,8)$ and $p_2=(1,4)$,
no holding cost, and zero terminal inventory. The relaxation orders
$(1,4)$ at cost $6$. Moving both orders to period 1 produces $(5,0)$,
whose first ending inventory is $4>B_1$. The original optimum costs $18$.
A rejected reconstruction does not imply that every possible reconstruction
fails, or that the relaxation necessarily has a strict optimality gap.

\subsection{Terminal policies and the oracle interface}
A zero-terminal oracle suffices. For prescribed $I_n=b$, replace the last
relaxed demand by $d_n+b$ and its final end-inventory bound by zero. Order
quantities and relaxed purchasing costs are unchanged under this mapping.
For free terminal inventory, append a period with demand $B_n$, zero
prices, zero holding cost, and zero ending inventory. Any admissible
original terminal inventory is completed by a free purchase. Appending
zero prices preserves monotonicity. Discard that dummy purchase after the
oracle call. These operations are applied to the zero-holding relaxation;
the constant $K$ in~\eqref{eq:adjusted} still refers to the original periods.

\subsection{A one-threshold certificate for multi-threshold instances}
When positive thresholds exist, let $Q_* =\min_{t,k\ge1}q_{t,k}$ and define
\begin{equation}
 g_t(x)=\begin{cases}
 p_{t,0}x,&x<Q_*,\\
 p_{t,m_t}x,&x\ge Q_*.
 \end{cases}\label{eq:coarsened}
\end{equation}
Then $g_t(x)\le f_t(x)$: below $Q_*$ no positive threshold has been reached,
and above it the surrogate uses the lowest local price, possibly too early.
If there are no positive thresholds, use identical tier prices and any
positive $Q_*$. Applying~\eqref{eq:adjusted}--\eqref{eq:prefix-prices} to
this surrogate yields another monotone one-threshold relaxation and a
lower bound $L$.

Reconstruct whole orders as before, but evaluate the candidate with the
\emph{original} thresholds. Capacity feasibility alone is no longer sufficient.
The appropriate acceptance test is
\begin{equation}
 x'\text{ satisfies all original constraints},\qquad C(x')=L.
 \label{eq:multi-certificate}
\end{equation}
Whenever~\eqref{eq:multi-certificate} holds, $L\le\OPT\le C(x')=L$ proves
optimality. With an $O(n\log n)$ oracle, this attempted certificate takes
$O(P+n\log n)$ operations, including reading and evaluating the tier data.

The extra cost test is essential. With two periods, $d=(0,2)$, $B=(10,10)$,
zero holding and terminal inventory, take local thresholds $(10,2)$ and
price pairs $((10,1),(9,2))$. The surrogate relaxation has value $2$.
Its two-unit order can be moved to period 1 without exceeding capacity,
but it does not reach that period's true threshold of ten. The actual
reconstructed cost is $20$, while the original optimum costs $4$.
The full test correctly rejects this candidate.

The complete certified procedure tries the relaxation, returns a candidate
only when its appropriate certificate succeeds, and otherwise invokes the
critical-grid or lattice algorithm. All worst-case bounds for the fallback
remain unchanged. No unconditional near-linear algorithm is asserted for
arbitrary temporal prices or multiple thresholds.

\section{Implementation and computational validation}\label{sec:validation}
\subsection{Reference implementations and reproducibility}
The supplementary package contains exact-rational Python implementations
of the one-threshold and multi-threshold critical-grid algorithms. The
multi-threshold implementation also provides the lattice algorithm and
an optional certificate wrapper. Rational data are represented by
\texttt{fractions.Fraction}. Threshold equality is resolved by selecting
the newly eligible lower tier. Grid and sum-set duplicates are removed
by sorted merging. An explicit state-count guard raises an exception if
its user-selected limit is exceeded; it never silently truncates the grid
or returns an approximate answer.

The critical-grid code computes a value using two layers and reconstructs
orders by the midpoint procedure of \ref{app:reconstruction}.
Every recovered plan is checked against the original inventory balances,
bounds, terminal policy, and objective. The lattice implementation uses
ordinary predecessor tables for plan recovery. The code does not include
the augmented-tree monotone algorithm: certificate tests use the exact
cubic implementation as a reference oracle. Consequently these tests
validate the reduction and certificate, not a near-linear implementation.

The results below were regenerated from the supplied scripts for this
manuscript. The package records pseudorandom seeds, generators, exact
inputs for the timing study, per-instance median timing records, machine-readable reports,
and environment information. Python 3.13.5 and SciPy 1.17.0 were used;
SciPy is required only for the independent linear-programming checks.

\subsection{Exact-answer checks}
We used two independently expressed comparison methods. First, an exhaustive
inventory-state dynamic program enumerates every feasible predecessor on
the specified small flow lattice and evaluates the actual all-units cost.
It does not use critical-state pruning or tierwise prefix minima.
Second, a tier-enumeration procedure solves one linear program for every
choice of actual tier intervals on small rational instances. Its variables
are order quantities, with cumulative inventory constraints. Tier intervals
are closed for these checks: closing an earlier tier at its upper threshold
can only overcharge the true lower cost there, while the correct adjacent
tier is separately enumerated. The minimum over all selections is thus
unchanged. This comparison does not use the critical grid.

\begin{table}[htbp]\centering\small
\caption{Regenerated validation checks. Counts describe tests, not claims
that every randomly generated instance is distinct.}
\label{tab:validation}
\begin{tabular}{p{85mm}rr}
\toprule
Check&\shortstack{One\\threshold}&\shortstack{Multiple/local\\thresholds}\\
\midrule
Integer-flow cases versus exhaustive inventory DP&1,000&1,000\\
Rational-lattice cases versus exhaustive inventory DP&333&500\\
Compatibility with one-threshold implementation&---&100\\
Certificate/fallback cases versus exact solver&333&200\\
Monotone cases with required certificate acceptance&100&---\\
Rational cases versus tier-enumeration LPs&100&60\\
Tier assignments enumerated in those LP tests&2,036&658\\
Named regression checks&6&6\\
\bottomrule
\end{tabular}
\end{table}

All exact-answer comparisons passed. All lattice instances in the
multi-threshold suite were also solved by both value-only and
plan-recovering variants of the lattice solver. The one-threshold suite
uses seed $251011368$, and the multi-threshold suite uses seed
$2510113682026$. The latter includes shared and period-specific thresholds,
zero-price ties, zero demands and capacities, and both prescribed and free
terminal inventory. The rational LP tests use independently generated
fractions rather than only rescaled integer instances. Their tolerance is
$10^{-7}\max\{1,|\text{LP objective}|\}$; all other objective comparisons use
exact rational equality. Complete parameter ranges appear in the scripts
and supplementary notes.

Of the 333 one-threshold certificate cases, 222 were accepted and 111 were
rejected; all 100 separately generated monotone cases were accepted.
Of the 200 multi-threshold certificate cases, 64 were accepted and 136
were rejected. Ninety-one of those rejections occurred without a capacity
violation, illustrating that objective equality under the original tiers
is an independent requirement. These are frequencies for the stated
small-instance generators, not estimates of acceptance rates in operational
data. Finite tests supplement, rather than replace, the proofs.

\subsection{An illustrative comparison of the two implementations}
A separate synthetic timing study compares objective-only critical-grid and
inventory-lattice solves. It uses $n\in\{20,40,80\}$, $m\in\{1,2,3\}$,
and three replications per combination of horizon, breakpoint count, and
capacity regime, giving 54 instances. The common thresholds are the first
$m$ values of $(7,13,23)$. Demands are independent integers from 1 to 8;
prices are independent integers from 1 to 30, sorted within each period;
and holding costs are integers from 0 to 3. Tight capacities are integers
from 0 to 6 and wide capacities from 15 to 30. Terminal inventory is zero.
The pseudorandom seed is $20261015$.

Each implementation is timed three times per instance after a warm-up;
method order alternates and the median elapsed time is retained.
Timing includes critical-grid construction where applicable, but excludes
instance generation, imports, and plan recovery. The tables aggregate those
per-instance medians geometrically. Both methods use the same exact-arithmetic
representation, and all 54 pairs of objective values agree exactly.

\begin{table}[htbp]\centering
\caption{Illustrative objective-only timings on 54 integer-flow instances. Each row pools nine instances; times are geometric means of per-instance median seconds. The ratio is critical-grid time divided by lattice time.}
\label{tab:timing}
\small
\begin{tabular}{clrrrrr}
\toprule
$m$ & Capacity & $M$ range & $W_{\max}$ & Grid (s) & Lattice (s) & Ratio\\
\midrule
1 & Tight & 74--396 & 7 & 0.0505 & 0.0018 & 28.34\\
1 & Wide & 70--381 & 31 & 0.0588 & 0.0122 & 4.82\\
2 & Tight & 74--375 & 7 & 0.1366 & 0.0020 & 68.73\\
2 & Wide & 99--385 & 31 & 0.1548 & 0.0148 & 10.48\\
3 & Tight & 68--384 & 7 & 0.1782 & 0.0020 & 90.16\\
3 & Wide & 67--386 & 31 & 0.2047 & 0.0172 & 11.87\\
\bottomrule
\end{tabular}
\end{table}

The run used Python 3.13.5 on Linux x86\_64 with an Intel Xeon Platinum 8370C processor (2.80\,GHz, shared container). The archived environment records the exact versions and per-instance timing records.

The experiment deliberately includes moderate integer storage capacities,
so the inventory-lattice width is small. It illustrates an implementation
choice rather than a claim of universal superiority of the critical grid.
On these data, repeated threshold sums coalesce; the experiment therefore
does not probe the worst-case $n^{m+1}$ growth of the critical-state set.
The code is a transparent Python reference implementation, not an optimized
solver, and shared-container timings should not be interpreted as stable
hardware performance estimates. No published implementation, commercial
mixed-integer optimizer, or near-linear oracle is benchmarked.

\section{Discussion and conclusions}\label{sec:conclusion}
For one common all-units threshold, arbitrary time-dependent prices can be
handled by an exact cubic dynamic program with quadratic working space.
The method uses a finite set of cumulative-purchase states supported by an
extreme-point argument and evaluates each period with two prefix minima.
It neither assumes continuity at a discount threshold nor permanently
prunes an inventory history from a comparison of marginal prices alone.
Arbitrary end-of-period bounds contribute a linear number of anchors.

With a fixed global set of $m$ threshold values, the same construction gives
$O_m(n^{m+2})$ time and $O_m(n^{m+1})$ working space, including a purchase plan.
This supplies a direct all-units implementation within the family of known
fixed-breakpoint polynomial bounds, not a priority claim for the exponent.
For unrestricted local thresholds, the method remains exact but its sum set
may be exponential. A bounded local tier count alone does not remove this
limitation. For moderate commensurate inventories, a pseudopolynomial
lattice implementation can be preferable.

The monotone-price relaxation provides an additional, logically separate
tool. It is useful whenever a reconstructed plan attains the relaxation's
lower bound, but its acceptance cannot be assumed with finite storage or
coarsened thresholds. The exact fallback provides the worst-case guarantee.
The reproducible computational checks support the implementations without
turning a conditional shortcut into an unconditional complexity result.

Several questions remain. Tighter critical-state sets may improve the
practical range of the fixed-$m$ method, and a smaller worst-case bound may
be possible for the one-threshold arbitrary-price case. Unrestricted
period-dependent thresholds require either a different polynomial-size
representation or a separate complexity classification; exponential subset
sums alone are not such a classification. Independent order limits, setup
charges, backlogging, and nonlinear holding costs change the model and
would require new statements and proofs rather than an automatic reuse
of the bounds reported here.

\section*{Data and code availability}
The Python source, synthetic instance generators, exact timing inputs,
per-instance median timing records, and regenerated validation reports accompany this
manuscript as supplementary material. No external operational dataset is
used. A persistent repository identifier should be added when the author
archives the final, reviewed version.

\section*{Declarations}
This manuscript is a work in progress and is subject to
further revisions, refinements, and corrections.
Artificial intelligence (AI) tools have been used
to assist in its preparation, including aspects of
writing, editing, and content development.
The author(s) retain full responsibility for the
accuracy, originality, and validity of the presented
material, including any errors or omissions.
The findings, interpretations, and conclusions
should be regarded as preliminary and not as
definitive results.
\section*{Funding and competing interests}
Funding information and the competing-interest declaration are to be
provided and confirmed by the author before submission.

\appendix
\renewcommand{\thetheorem}{\Alph{section}.\arabic{theorem}}
\section{Space-efficient recovery of an optimal plan}\label{app:reconstruction}
This appendix proves the recovery claim without storing all $n$ value
layers. The construction uses the forward--backward divide-and-conquer
principle associated with \citet{hirschberg1975}. The additional argument
here is the partition of cumulative-coordinate widths, which prevents an
extra logarithmic factor in the time bound.

\subsection{Forward and backward evaluations}
A recovery subproblem specifies boundary times $a<b$ and endpoint grid
indices $\ell\le r$, meaning $Y_a=v_\ell$ and $Y_b=v_r$.
Cumulative purchasing is nondecreasing, so all intermediate coordinates
belong to $\{v_\ell,\ldots,v_r\}$. Forward layers use
\eqref{eq:dp-multi} on this restricted grid and initialize cost zero only
at $v_\ell$. Their holding costs are charged at the ending boundary of
each transition.

For backward evaluation, let $B_t(u)$ be the minimum cost in periods
$t+1,\ldots,b$, conditional on $Y_t=u$ and $Y_b=v_r$. Initialize
$B_b(v_r)=0$ and all other values to $+\infty$. The backward recurrence is
\begin{align}
 B_{t-1}(u)=\min_k\left\{-p_{t,k}u+
 \min_{\substack{v\in\V\cap[v_\ell,v_r]\\
                 v\ge u+q_{t,k},\ D_t\le v\le\U_t}}
 \bigl(B_t(v)+p_{t,k}v+h_t(v-D_t)\bigr)\right\}.
 \label{eq:backward}
\end{align}
It is evaluated only at feasible source coordinates. For a fixed tier,
scan $u$ in decreasing order and maintain the minimum over the expanding
suffix $v\ge u+q_{t,k}$. A backward layer consequently has the same
linear-per-tier complexity as a forward layer.

\subsection{Midpoint selection and recursion}
When $a<b-1$, choose $c=\lfloor(a+b)/2\rfloor$. Compute forward costs
$A_c(v)$ for periods $a+1,\ldots,c$ and backward costs $B_c(v)$ for
periods $c+1,\ldots,b$. Choose a minimizer
\begin{equation}
 v_s\in\argmin_{v\in\V\cap[v_\ell,v_r]}
                 \{A_c(v)+B_c(v)\}.\label{eq:midpoint}
\end{equation}
There is no duplicated holding charge because the backward cost starts
strictly after period $c$. Release both value arrays, then solve the
subproblems $(a,c,\ell,s)$ and $(c,b,s,r)$.
When $b=a+1$, set $x_b=v_r-v_\ell$. When $\ell=r$, all orders in the
subproblem are zero, so no further dynamic program is needed. Initialize
the complete output vector to zero once before recursing. Feasibility of
a selected subproblem follows from the finite costs in~\eqref{eq:midpoint}.
At the root, its terminal index is selected by the main forward solve.

\begin{proposition}[Recovery complexity]\label{prop:recovery}
For a grid of size $M$, recovery takes $O(nK_{\max}M)$ operations and
$O(M+n)$ auxiliary space. For the one-threshold case this is $O(nM)$ time.
\end{proposition}
\begin{proof}
Correctness follows by splitting an optimal path at~\eqref{eq:midpoint}.
Any cheaper child path would improve the parent path, and the concatenation
of two minimizing child paths attains the parent's optimum.

For the running time, define a subproblem's coordinate width as $w=r-\ell$,
not its number of grid points. The two child widths sum to $w$.
Zero-width subproblems require no value-layer computation. For positive
width, the number of processed coordinates is $w+1\le2w$. At every
recursion depth, the sum of the positive widths is at most the root width,
which is at most $M-1$. A time interval at depth $d$ has length at most
$\lceil n/2^d\rceil$. Thus the value-layer work at that depth is
\[
 O\!\left(K_{\max}M\left\lceil\frac{n}{2^d}\right\rceil\right).
\]
Summing over $d=0,\ldots,\lceil\log_2 n\rceil$ gives
$O(nK_{\max}M)$ because the sum of these ceilings is $O(n)$.
The recursion has $O(n)$ nodes and writes at most $n$ purchases.

Only the arrays for the current subproblem need be retained: parent
arrays are released before either child is entered. They occupy $O(M)$
space, the output occupies $O(n)$, and depth-first recursion uses
$O(\log(n+1))$ additional bookkeeping. The original sorted grid is shared
by all subproblems.
\end{proof}



\begin{thebibliography}{11}
\expandafter\ifx\csname natexlab\endcsname\relax\def\natexlab#1{#1}\fi
\providecommand{\url}[1]{\texttt{#1}}
\providecommand{\href}[2]{#2}
\providecommand{\path}[1]{#1}
\providecommand{\DOIprefix}{doi:}
\providecommand{\ArXivprefix}{arXiv:}
\providecommand{\URLprefix}{URL: }
\providecommand{\Pubmedprefix}{pmid:}
\providecommand{\doi}[1]{\href{http://dx.doi.org/#1}{\path{#1}}}
\providecommand{\Pubmed}[1]{\href{pmid:#1}{\path{#1}}}
\providecommand{\bibinfo}[2]{#2}
\ifx\xfnm\relax \def\xfnm[#1]{\unskip,\space#1}\fi
\bibitem[{Down et~al.(2021)Down, Karakostas, Kolliopoulos and
  Rostami}]{down2021}
\bibinfo{author}{Down, D.G.}, \bibinfo{author}{Karakostas, G.},
  \bibinfo{author}{Kolliopoulos, S.G.}, \bibinfo{author}{Rostami, S.},
  \bibinfo{year}{2021}.
\newblock \bibinfo{title}{Single-item lot-sizing with quantity discount and
  bounded inventory}.
\newblock \bibinfo{journal}{Operations Research Letters} \bibinfo{volume}{49},
  \bibinfo{pages}{877--882}.
\newblock \DOIprefix\doi{10.1016/j.orl.2021.11.002}.
\bibitem[{Hirschberg(1975)}]{hirschberg1975}
\bibinfo{author}{Hirschberg, D.S.}, \bibinfo{year}{1975}.
\newblock \bibinfo{title}{A linear space algorithm for computing maximal common
  subsequences}.
\newblock \bibinfo{journal}{Communications of the ACM} \bibinfo{volume}{18},
  \bibinfo{pages}{341--343}.
\newblock \DOIprefix\doi{10.1145/360825.360861}.
\bibitem[{Koca et~al.(2014)Koca, Yaman and Akt{\"u}rk}]{koca2014}
\bibinfo{author}{Koca, E.}, \bibinfo{author}{Yaman, H.},
  \bibinfo{author}{Akt{\"u}rk, M.S.}, \bibinfo{year}{2014}.
\newblock \bibinfo{title}{Lot sizing with piecewise concave production costs}.
\newblock \bibinfo{journal}{INFORMS Journal on Computing} \bibinfo{volume}{26},
  \bibinfo{pages}{767--779}.
\newblock \DOIprefix\doi{10.1287/ijoc.2014.0597}.
\bibitem[{Love(1973)}]{love1973}
\bibinfo{author}{Love, S.F.}, \bibinfo{year}{1973}.
\newblock \bibinfo{title}{Bounded production and inventory models with
  piecewise concave costs}.
\newblock \bibinfo{journal}{Management Science} \bibinfo{volume}{20},
  \bibinfo{pages}{313--318}.
\newblock \DOIprefix\doi{10.1287/mnsc.20.3.313}.
\bibitem[{Ou(2017)}]{ou2017}
\bibinfo{author}{Ou, J.}, \bibinfo{year}{2017}.
\newblock \bibinfo{title}{Improved exact algorithms to economic lot-sizing with
  piecewise linear production costs}.
\newblock \bibinfo{journal}{European Journal of Operational Research}
  \bibinfo{volume}{256}, \bibinfo{pages}{777--784}.
\newblock \DOIprefix\doi{10.1016/j.ejor.2016.06.040}.
\bibitem[{Papadopoulos(2026a)}]{papmonge2026}
\bibinfo{author}{Papadopoulos, K.}, \bibinfo{year}{2026}a.
\newblock \bibinfo{title}{Monge matrix searching for lot sizing with
  piecewise-concave production costs}.
\newblock \URLprefix \url{https://arxiv.org/abs/2610.03277v1}.
  \bibinfo{note}{arXiv:2610.03277v1. Preprint; 2 October 2026}.
\bibitem[{Papadopoulos(2026b)}]{pappl2026}
\bibinfo{author}{Papadopoulos, K.}, \bibinfo{year}{2026}b.
\newblock \bibinfo{title}{A novel exact algorithm for economic lot-sizing with
  piecewise linear production costs}.
\newblock \bibinfo{journal}{Journal of Global Optimization}
  \bibinfo{volume}{96}, \bibinfo{pages}{171--205}.
\newblock \DOIprefix\doi{10.1007/s10898-026-01636-3}.
\bibitem[{Papadopoulos(2026c)}]{papmono2026}
\bibinfo{author}{Papadopoulos, K.}, \bibinfo{year}{2026}c.
\newblock \bibinfo{title}{An {$O(n\log n)$} algorithm for single-item lot
  sizing with a one-breakpoint all-units discount and non-increasing prices}.
\newblock \URLprefix \url{https://arxiv.org/abs/2510.11368v4}.
  \bibinfo{note}{arXiv:2510.11368v4. Preprint; version revised 5 March 2026}.
\bibitem[{Shaw and Wagelmans(1998)}]{shaw1998}
\bibinfo{author}{Shaw, D.X.}, \bibinfo{author}{Wagelmans, A.P.M.},
  \bibinfo{year}{1998}.
\newblock \bibinfo{title}{An algorithm for single-item capacitated economic lot
  sizing with piecewise linear production costs and general holding costs}.
\newblock \bibinfo{journal}{Management Science} \bibinfo{volume}{44},
  \bibinfo{pages}{831--838}.
\newblock \DOIprefix\doi{10.1287/mnsc.44.6.831}.
\bibitem[{Swoveland(1975)}]{swoveland1975}
\bibinfo{author}{Swoveland, C.}, \bibinfo{year}{1975}.
\newblock \bibinfo{title}{A deterministic multi-period production planning
  model with piecewise concave production and holding-backorder costs}.
\newblock \bibinfo{journal}{Management Science} \bibinfo{volume}{21},
  \bibinfo{pages}{1007--1013}.
\newblock \DOIprefix\doi{10.1287/mnsc.21.9.1007}.
\bibitem[{Wagner and Whitin(1958)}]{wagner1958}
\bibinfo{author}{Wagner, H.M.}, \bibinfo{author}{Whitin, T.M.},
  \bibinfo{year}{1958}.
\newblock \bibinfo{title}{Dynamic version of the economic lot size model}.
\newblock \bibinfo{journal}{Management Science} \bibinfo{volume}{5},
  \bibinfo{pages}{89--96}.
\newblock \DOIprefix\doi{10.1287/mnsc.5.1.89}.

\end{thebibliography}
\end{document}